\documentclass{article}

\usepackage{times}
\usepackage{amssymb,amsmath,amsthm}
\usepackage{mathrsfs}
\usepackage{epsfig}
\usepackage{color}

\newcommand\ie{{\em i.e.}}

\newcommand\cf{{\em cf.}~}

\def\B{\mathscr B}
\def\C{\mathbb C}
\def\d{\mathrm{d}}
\def\F{\mathscr F}
\def\G{\mathcal G}
\def\GG{\mathscr G}
\def\H{\mathcal H}

\def\K{\mathscr K}
\def\LL{\mathcal L}
\def\N{\mathbb N}
\def\R{\mathbb R}
\def\S{\mathbb S}
\def\SS{\mathscr S}
\def\U{\mathscr U}
\def\Hrond{\mathscr H}
\def\HS{\mathfrak h}
\def\Pv{\mathrm{Pv}}

\def\dom{\mathcal D}
\def\lone{\mathsf{L}^{\:\!\!1}}

\def\ltwo{\mathsf{L}^{\:\!\!2}}
\def\linf{\mathsf{L}^{\:\!\!\infty}}
\def\e{\mathop{\mathrm{e}}\nolimits}

\def\slim{\mathop{\hbox{\rm s-}\lim}\nolimits}

\newtheorem{Theorem}{Theorem}[section]
\newtheorem{Remark}[Theorem]{Remark}
\newtheorem{Lemma}[Theorem]{Lemma}

\begin{document}


\title{New expressions for the wave operators of Schr\"odinger operators in $\R^3$}

\author{S. Richard$^1$\footnote{This work has been done during the stay of S. Richard
in Japan and has been supported by the Japan Society for the Promotion of Science
(JSPS) and by ``Grants-in-Aid for scientific Research''.}~~and R. Tiedra de
Aldecoa$^2$\footnote{Supported by the Chilean Fondecyt Grant 1090008
and by the Iniciativa Cientifica Milenio ICM P07-027-F ``Mathematical Theory of
Quantum and Classical Magnetic Systems'' from the Chilean Ministry of Economy.}}

\date{\small}
\maketitle \vspace{-1cm}

\begin{quote}
\emph{
\begin{itemize}
\item[$^1$] Universit\'e de Lyon; Universit\'e
Lyon 1; CNRS, UMR5208, Institut Camille Jordan, \\
43 blvd du 11 novembre 1918, F-69622
Villeurbanne-Cedex, France.
\item[$^2$] Facultad de Matem\'aticas, Pontificia Universidad Cat\'olica de Chile,\\
Av. Vicu\~na Mackenna 4860, Santiago, Chile
\item[] \emph{E-mails:} richard@math.univ-lyon1.fr, rtiedra@mat.puc.cl
\end{itemize}
}
\end{quote}


\begin{abstract}
We prove new and explicit formulas for the wave operators of Schr\"odinger operators
in $\R^3$. These formulas put into light the very special role played by the
generator of dilations and validate the topological approach of Levinson's theorem
introduced in a previous publication. Our results hold for general (not spherically
symmetric) potentials decaying fast enough at infinity, without any assumption on the
absence of eigenvalue or resonance at $0$-energy.
\end{abstract}


\smallskip


\section{Introduction and main theorem}
\setcounter{equation}{0}

The purpose of this work is to establish explicit and completely new expressions for
the wave operators of Schr\"odinger operators in $\R^3$, and as a by-product to
validate the use of the topological approach of Levinson's theorem.

The set-up is the standard one. We consider in the Hilbert space $\H:=\ltwo(\R^3)$
the free Schr\"odinger operator $H_0:=-\Delta$ and the perturbed Schr\"odinger
operator $H:=-\Delta+V$, with $V$ a measurable bounded real function on $\R^3$
decaying fast enough at infinity. In such a situation, it is well-known that the
wave operators
\begin{equation}\label{wave}
W_\pm:=\slim_{t\to\pm\infty}\e^{itH}\e^{-itH_0}
\end{equation}
exist and are asymptotically complete \cite{AJS,Pea88,RS3}, and as a consequence that
the scattering operator $S:=W_+^*W_-$ is a unitary operator in $\H$. Moreover, it is
also well-known that one can write time-independent expressions for $W_\pm$ by using
the stationary formulation of scattering theory (see \cite{BW83,Kur78,Yaf92}).

Among the many features of the wave operators, their mapping properties between
weighted Hilbert spaces, weighted Sobolev spaces and $\mathsf{L}^{\:\!\!p}$-spaces
have attracted a lot of attention (see for instance the seminal papers
\cite{JN94,Yaj95,Yaj97,Yaj06} and the preprint \cite{Bec11} which contains an
interesting historical overview and many references). Also, recent technics developed
for the study of the wave operators have been used to obtain dispersive estimates for
Schr\"odinger operators \cite{Bec12,ES04,ES06,Yaj05}. Our point here, which can be
inscribed in this line of general works on wave operators, is to show that the
time-independent expressions for $W_\pm$ can be made completely explicit, up to a
compact term. Namely, if $\B(\H)$ (resp. $\K(\H)$) denotes the set of bounded (resp.
compact) operators in $\H$, and if $A$ stands for the generator of dilations in
$\R^3$, then we have the following result:

\begin{Theorem}\label{Java}
Let $V$ satisfy $|V(x)|\le{\rm Const.}\;\!(1+|x|)^{-\sigma}$ with $\sigma>7$ for
almost every $x\in\R^3$. Then, one has in $\B(\H)$ the equalities
\begin{equation}\label{jolieformule}
W_-=1+R(A)(S-1)+K
\qquad\hbox{and}\qquad
W_+=1+\big(1-R(A)\big)(S^*-1)+K',
\end{equation}
with $R(A):=\frac12\big(1+\tanh(\pi A)-i\cosh(\pi A)^{-1}\big)$ and $K,K'\in\K(\H)$.
\end{Theorem}

We stress that the absence of eigenvalue or resonance at $0$-energy is not assumed.
On the other hand, if such an implicit hypothesis is made, then the same result holds
under a weaker assumption on the decay of $V$ at infinity. We also note that no
spherical symmetry is imposed on $V$.

Our motivation for proving Theorem \ref{Java} was the observation made in \cite{KR06}
(and applied to various situations in \cite{BS12,IR12,KR08,PR11,RT10}) that
Levinson's theorem can be reinterpreted as an index theorem, with a proof based on an
explicit expression for the wave operators. The main idea is to show that the wave
operators belong to a certain $C^*$-algebra. Once such an affiliation property is
settled, the machinery of non-commutative topology leads naturally to an index
theorem. In its original form, this index theorem corresponds to Levinson's theorem;
that is, the equality between the number of bound states of the operator $H$ and an
expression (trace) involving the scattering operator $S$. For more complex scattering
systems, other topological equalities involving higher degree traces can also be
derived (see \cite{KPR} for more explanations).

For the scattering theory of Schr\"odinger operators in $\R^3$, the outcomes of this
topological approach have been detailed in \cite{KR12}: It has been shown how
Levinson's theorem can be interpreted as an index theorem, and how one can derive
from it various formulas for the number of bound states of $H$ in terms of the
scattering operator and a second operator related to the $0$-energy. However, a
technical argument was missing, and an implicit assumption had to be made
accordingly. Theorem \ref{Java} makes this implicit assumption no longer necessary,
and thus allows one to apply all the results of \cite{KR12} (see Remark
\ref{Rem_comments} for some more comments).

Let us now present a more detailed description of our results. As mentioned above,
our goal was to obtain an explicit formula for the wave operators, as required by the
$C^*$-algebras framework. However, neither the time dependant formula \eqref{wave},
nor the stationary approach as presented for instance in \cite{Yaf92}, provided us
with a sufficiently precise answer. This motivated us to show in Theorem
\ref{BigMama} of Section \ref{Sec_one} that the difference $W_--1$ is unitarily
equivalent to a product of three explicit bounded operators. The result is exact and
no compact operator as in the statement of Theorem \ref{Java} has to be added. In
addition, each of the three operators is either an operator of multiplication by an
operator-valued function, or a simple function of the generator of dilation in
$\ltwo(\R_+)$. For these reasons, we expect that the formula of Theorem \ref{BigMama}
might have various applications, as for example for the mapping properties of $W_-$.
Finally, the commutation of two of the three operators reveals the presence of the
scattering operator up to a compact term, as stated in Theorem \ref{Java}. One
deduces from this new expression for $W_-$ the corresponding expression for $W_+$.

As a conclusion, we emphasize once more that the present work validates the use of
the topological approach of Levinson's theorem, as presented in \cite{KR12}. It also
implicitly shows that this $C^*$-algebraic approach of scattering theory leads to new
questions and new results, as exemplified by the explicit formula presented in
Theorem \ref{Java}.\\

\noindent
{\bf Notations\hspace{1pt}:}
$\N:=\{0,1,2,\ldots\}$ is the set of natural numbers, $\R_+:=(0,\infty)$, and $\SS$
is the Schwartz space on $\R^3$. The sets $\H^s_t$ are the weighted Sobolev spaces
over $\R^3$ with index $s\in\R$ associated to derivatives and index $t\in\R$
associated to decay at infinity \cite[Sec.~4.1]{ABG} (with the convention that
$\H^s:=\H^s_0$ and $\H_t:=\H^0_t$). The three-dimensional Fourier transform $\F$ is a
topological isomorphism of $\H^s_t$ onto $\H^t_s$ for any $s,t\in\R$. Given two
Banach spaces $\G_1$ and $\G_2$, $\B(\G_1,\G_2)$ (resp. $\K(\G_1,\G_2)$) stands for
the set of bounded (resp. compact) operators from $\G_1$ to $\G_2$. Finally,
$\otimes$ (resp. $\odot$) stands for the closed (resp. algebraic) tensor product of
Hilbert spaces or of operators.

\section{New expressions for the wave operators}\label{Sec_one}
\setcounter{equation}{0}

We start by introducing the Hilbert spaces we use throughout the paper; namely,
$\H:=\ltwo(\R^3)$, $\HS:=\ltwo(\S^2)$ and $\Hrond:=\ltwo\big(\R_+;\HS\big)$ with
respective scalar product $\langle\;\!\cdot\;\!,\;\!\cdot\;\!\rangle$ and norm
$\|\;\!\cdot\;\!\|$ indexed accordingly. The Hilbert space $\Hrond$ hosts the
spectral representation of the operator $H_0=-\Delta$ with domain $\dom(H_0)=\H^2$,
\ie, there exists a unitary operator $\F_0:\H\to\Hrond$ satisfying
$$
(\F_0H_0f)(\lambda)
=\lambda\;\!(\F_0 f)(\lambda)
\equiv(L\;\!\F_0 f)(\lambda),
\quad f\in\dom(H_0),\hbox{ a.e. }\lambda\in\R_+,
$$
with $L$ the maximal multiplication operator in $\Hrond$ by the variable in $\R_+$.
The explicit formula for $\F_0$ is
\begin{equation}\label{def_F_0}
\big((\F_0 f)(\lambda)\big)(\omega)
=\textstyle\big(\frac\lambda4\big)^{1/4}(\F f)(\sqrt\lambda\;\!\omega)
=\textstyle\big(\frac\lambda4\big)^{1/4}
\big(\gamma(\sqrt\lambda)\;\!\F f\big)(\omega),
\quad f\in\SS,~\lambda\in\R_+,~\omega\in\S^2,
\end{equation}
with $\gamma(\lambda):\SS\to\HS$ the trace operator given by
$\big(\gamma(\lambda)f\big)(\omega):=f(\lambda\;\!\omega)$.

The potential $V\in\linf(\R^3;\R)$ of the perturbed Hamiltonian $H:=H_0+V$ satisfies
for some $\sigma>0$ the condition
\begin{equation}\label{condV}
|V(x)|\le{\rm Const.}\;\!\langle x\rangle^{-\sigma},\quad\hbox{a.e. }x\in\R^3,
\end{equation}
with $\langle x\rangle:=\sqrt{1+x^2}$. Since $V$ is bounded, $H$ is self-adjoint with
domain $\dom(H)=\dom(H_0)$. Also, it is well-known \cite[Thm.~12.1]{Pea88} that the
wave operators defined by \eqref{wave} exist and are asymptotically complete if
$\sigma>1$. In stationary scattering theory one defines the wave operators in terms
of suitable limits of the resolvents of $H_0$ and $H$ on the real axis. We shall
mainly use this second approach, noting that for this model both definitions for the
wave operators do coincide (see \cite[Sec.~5.3]{Yaf92}).

Now, we recall from \cite[Eq.~2.7.5]{Yaf92} that for suitable $f,g\in\H$ the
stationary expressions for the wave operators are given by
$$
\big\langle W_\pm f,g\big\rangle_\H
=\int_\R\d\lambda\,\lim_{\varepsilon\searrow0}\frac\varepsilon\pi
\big\langle R_0(\lambda\pm i\varepsilon)f,
R(\lambda\pm i\varepsilon)g\big\rangle_\H\;\!,
$$
where $R_0(z):=(H_0-z)^{-1}$ and $R(z):=(H-z)^{-1}$, $z\in\C\setminus\R$, are the
resolvents of the operators $H_0$ and $H$. We also recall from \cite[Sec.~1.4]{Yaf92}
that the limit
$
\lim_{\varepsilon\searrow 0}
\big\langle\delta_\varepsilon(H_0-\lambda)f,g\big\rangle_\H
$
with
$
\delta_\varepsilon(H_0-\lambda)
:=\frac\varepsilon\pi R_0(\lambda\mp i\varepsilon)\;\!R_0(\lambda\pm i\varepsilon)
$
exists for a.e. $\lambda\in\R$ and that
$$
\big\langle f,g\big\rangle_\H
=\int_\R\d\lambda\,\lim_{\varepsilon\searrow0}
\big\langle\delta_\varepsilon(H_0-\lambda)f,g\big\rangle_\H\;\!.
$$
Thus, taking into account the second resolvent equation, one infers that
\begin{equation*}
\big\langle(W_\pm-1)f,g\big\rangle_\H
=-\int_\R\d\lambda\,\lim_{\varepsilon\searrow0}\big\langle
\delta_\varepsilon(H_0-\lambda)f,\big(1+VR_0(\lambda\pm i\varepsilon)\big)^{-1}\;\!
VR_0(\lambda\pm i\varepsilon)g\big\rangle_\H\;\!.
\end{equation*}

We now derive new expressions for the wave operators in the spectral representation
of $H_0$; that is, for the operators $\F_0(W_\pm-1)\F_0^*$. So, let
$\varphi,\psi$ be suitable elements of $\Hrond$ (precise conditions will be specified
in Theorem \ref{BigMama} below), then one obtains that
\begin{align*}
&\big\langle\F_0(W_\pm-1)\F_0^*\varphi,\psi\big\rangle_{\!\Hrond}\\
&=-\int_\R\d\lambda\,\lim_{\varepsilon\searrow0}
\big\langle V\big(1+R_0(\lambda\mp i\varepsilon)V\big)^{-1}\F_0^*\;\!
\delta_\varepsilon(L-\lambda)\varphi,
\F_0^*\;\!(L-\lambda\mp i\varepsilon)^{-1}\psi\big\rangle_\H\\
&=-\int_\R\d\lambda\,\lim_{\varepsilon\searrow0}\int_0^\infty\d\mu\,
\big\langle\big\{\F_0V\big(1+R_0(\lambda\mp i\varepsilon)V\big)^{-1}\F_0^*\;\!
\delta_\varepsilon(L-\lambda)\varphi\big\}(\mu),(\mu-\lambda\mp i\varepsilon)^{-1}
\psi(\mu)\big\rangle_\HS.
\end{align*}
Using the short hand notation $T(z):=V\big(1+R_0(z)V\big)^{-1}$, $z\in\C\setminus\R$,
one thus gets the equality
\begin{align}
&\big\langle\;\!\F_0(W_\pm-1)\F_0^*\varphi,\psi\big\rangle_{\!\Hrond}\nonumber\\
&=-\int_\R\d\lambda\,\lim_{\varepsilon\searrow0}\int_0^\infty\d\mu\,
\big\langle\big\{\F_0\;\!T(\lambda\mp i\varepsilon)\;\!\F_0^*\;\!
\delta_\varepsilon(L-\lambda)\varphi\big\}(\mu),
(\mu-\lambda\mp i\varepsilon)^{-1}\psi(\mu)\big\rangle_\HS.\label{start}
\end{align}

The next step is to exchange the integral over $\mu$ and the limit
$\varepsilon\searrow0$ in the previous expression. To do it properly, we need a
series of preparatory lemmas. First of all, we recall that for $\lambda>0$ the trace
operator $\gamma(\lambda)$ extends to an element of $\B(\H^s_t,\HS)$ for each $s>1/2$
and $t\in\R$ and that the map $\R_+\ni\lambda\mapsto\gamma(\lambda)\in\B(\H^s_t,\HS)$
is continuous \cite[Sec.~3]{Jen81}. As a consequence, the operator
$\F_0(\lambda):\SS\to\HS$ given by $\F_0(\lambda)f:=(\F_0f)(\lambda)$ extends to an
element of $\B(\H^s_t,\HS)$ for each $s\in\R$ and $t>1/2$, and the map
$\R_+\ni\lambda\mapsto\F_0(\lambda)\in\B(\H^s_t,\HS)$ is continuous.

We shall now strengthen these standard results.

\begin{Lemma}\label{lem1}
Let $s\ge 0$ and $t>3/2$. Then, the functions
$$
(0,\infty)\ni\lambda\mapsto\lambda^{\pm1/4}\F_0(\lambda)\in\B(\H^s_t,\HS)
$$
are continuous and bounded.
\end{Lemma}

\begin{proof}
The continuity of the functions
$
(0,\infty)\ni\lambda\mapsto\lambda^{\pm1/4}\F_0(\lambda)\in\B(\H^s_t,\HS)
$
follows from what has been said before. For the boundedness, it is sufficient to show
that the map $\lambda\mapsto\lambda^{-1/4}\|\F_0(\lambda)\|_{\B(\H^s_t,\HS)}$ is
bounded in a neighbourhood of $0$, and that the map
$\lambda\mapsto\lambda^{1/4}\|\F_0(\lambda)\|_{\B(\H^s_t,\HS)}$ is bounded in a
neighbourhood of $+\infty$. The first bound follows from the asymptotic development
for small $\lambda>0$ of the operator $\gamma(\sqrt\lambda)\;\!\F\in\B(\H^s_t,\HS)$
(see \cite[Sec.~5]{JK79}) and the second bound follows from \cite[Thm.~1.1.4]{Yaf10}
which implies that the map
$\lambda\mapsto\lambda^{1/4}\|\F_0(\lambda)\|_{\B(\H^s_t,\HS)}$ is bounded on $\R_+$.
Note that only the case $s=0$ is presented in \cite[Thm.~1.1.4]{Yaf10}, but the
extension to the case $s\ge0$ is trivial since $\H^s_t\subset\H^0_t$ for any $s>0$.
\end{proof}

One immediately infers from Lemma \ref{lem1} that the function
$\R_+\ni\lambda\mapsto\|\F_0(\lambda)\|_{\B(\H^s_t,\HS)}\in\R$ is continuous and
bounded for any $s\ge0$ and $t>3/2$. Also, one can strengthen the statement of Lemma
\ref{lem1} in the case of the minus sign\;\!:

\begin{Lemma}\label{lem2}
Let $s >-1$ and $t>3/2$. Then, $\F_0(\lambda)\in\K(\H^s_t,\HS)$ for each
$\lambda\in\R_+$, and the function
$\R_+\ni\lambda\mapsto\lambda^{-1/4}\F_0(\lambda)\in\K(\H^s_t,\HS)$ is continuous,
admits a limit as $\lambda\searrow0$ and vanishes as $\lambda\to\infty$.
\end{Lemma}

\begin{proof}
The inclusion $\F_0(\lambda)\in\K(\H^s_t,\HS)$ follows from the compact embedding
$\H^s_t\subset \H^{s'}_{t'}$ for any $s'<s$ and $t'<t$ (see for instance
\cite[Prop.~4.1.5]{ABG}).

For the continuity and the existence of the limit as $\lambda\searrow0$ one can use
the same argument as the one used in the proof of Lemma \ref{lem1}. For the limit as
$\lambda\to\infty$, we define the regularizing operator
$\langle P\rangle^{-s}:=(1-\Delta)^{-s/2}$ and then observe that
$
\lambda^{-1/4}\F_0(\lambda)\langle P\rangle^{-s}
=\lambda^{-1/4}(1+\lambda)^{-s/2}\F_0(\lambda)
$
for each $\lambda\in\R_+$ (see \eqref{def_F_0}). It follows that
$\lim_{\lambda\to\infty}\|\lambda^{-1/4}\F_0(\lambda)\|_{\B(\H^s_t,\HS)}=0$ if and
only if
$
\lim_{\lambda\to\infty}
\|\lambda^{-1/4}(1+\lambda)^{-s/2}\F_0(\lambda)\|_{\B(\H_t,\HS)}=0
$.
So, the claim follows from Lemma \ref{lem1} (with the positive sign) as long as
$-1/4-s/2<1/4$, which is equivalent to the condition $s>-1$.
\end{proof}

From now on, we use the notation $C_{\rm c}(\R_+;\G)$ for the set of compactly
supported and continuous functions from $\R_+$ to some Hilbert space $\G$. With this
notation and what precedes, we note that the multiplication operator
$M:C_{\rm c}(\R_+;\H^s_t)\to\Hrond$ given by
\begin{equation}\label{defdeM}
(M\xi)(\lambda):=\lambda^{-1/4}\F_0(\lambda)\;\!\xi(\lambda),
\quad\xi\in C_{\rm c}(\R_+;\H^s_t),~\lambda\in\R_+,
\end{equation}
extends for $s\ge 0$ and $t>3/2$ to an element of
$\B\big(\ltwo(\R_+;\H^s_t),\Hrond\big)$.

The next step is to deal with the limit $\varepsilon\searrow0$ of the operator
$\delta_\varepsilon(L-\lambda)$ in Equation \eqref{start}. For that purpose, we shall
use the continuous extension of the scalar product
$\langle\;\!\cdot\;\!,\;\!\cdot\;\!\rangle_\H$ to a duality
$\langle\;\!\cdot\;\!,\;\!\cdot\;\!\rangle_{\H^s_t,\H^{-s}_{-t}}$ between $\H^s_t$
and $\H^{-s}_{-t}$.

\begin{Lemma}\label{lemlimite}
Take $s\ge0$, $t>3/2$, $\lambda\in\R_+$ and $\varphi\in C_{\rm c}(\R_+;\HS)$. Then,
we have
$$
\lim_{\varepsilon\searrow 0}\big\|\F_0^*\;\!\delta_\varepsilon(L-\lambda)\varphi
-\F_0(\lambda)^*\varphi(\lambda)\big\|_{\H^{-s}_{-t}}=0.
$$
\end{Lemma}

\begin{proof}
By definition of the norm of $\H^{-s}_{-t}$, one has
\begin{align}
&\big\|\F_0^*\;\!\delta_\varepsilon(L-\lambda)\varphi
-\F_0(\lambda)^*\varphi(\lambda)\big\|_{\H^{-s}_{-t}}\nonumber\\
&=\sup_{f\in\SS,\,\|f\|_{\H^s_t=1}}
\Big|\big\langle f,\F_0^*\;\!\delta_\varepsilon(L-\lambda)\varphi
-\F_0(\lambda)^*\varphi(\lambda)\big\rangle_{\H^s_t,\H^{-s}_{-t}}\Big|\nonumber\\
&=\sup_{f\in\SS,\,\|f\|_{\H^s_t}=1}
\left|\frac1\pi\int_0^\infty\d\mu\,\left\langle\F_0(\mu)f,
\frac\varepsilon{(\mu-\lambda)^2+\varepsilon^2}\;\!\varphi(\mu)\right\rangle_\HS
-\big\langle\F_0(\lambda)f,\varphi(\lambda)\big\rangle_\HS\right|\nonumber\\
&\le\sup_{f\in\SS,\,\|f\|_{\H^s_t}=1}
\left|\frac1\pi\int_0^\infty\d\mu\,\left\langle\big(\F_0(\mu)-\F_0(\lambda)\big)f,
\frac\varepsilon{(\mu-\lambda)^2+\varepsilon^2}\;\!\varphi(\mu)\right\rangle_\HS
\right|\label{h1}\\
&\quad+\sup_{f\in\SS,\,\|f\|_{\H^s_t}=1}
\left|\frac1\pi\int_0^\infty\d\mu\,\left\langle\F_0(\lambda)f,
\frac\varepsilon{(\mu-\lambda)^2+\varepsilon^2}
\big(\varphi(\mu)-\varphi(\lambda)\big)\right\rangle_\HS\right|\label{h2}\\
&\quad+\sup_{f\in\SS,\,\|f\|_{\H^s_t}=1}
\left|\frac1\pi\int_0^\infty\d\mu\,\left\langle\F_0(\lambda)f,
\frac\varepsilon{(\mu-\lambda)^2+\varepsilon^2}\;\!\varphi(\lambda)
\right\rangle_\HS
-\big\langle\F_0(\lambda)f,\varphi(\lambda)\big\rangle_\HS\right|\label{h3}
\end{align}
Clearly, the term \eqref{h3} converges to $0$ as $\varepsilon\searrow0$, as expected.
Furthermore, the term \eqref{h1} converges to $0$ as $\varepsilon\searrow0$ because
of the continuity and the boundedness of the function
$\lambda\mapsto\|\F_0(\lambda)\|_{\B(\H^s_t,\HS)}$ (mentioned just after Lemma
\ref{lem1}) together with the boundedness of the map
$\lambda\mapsto\|\varphi(\lambda)\|_\HS$. Finally, the term \eqref{h2} also converges
to $0$ as $\varepsilon \searrow 0$ because of the continuity and the boundedness of
the function $\lambda\mapsto \varphi(\lambda)\in \HS$ together with the boundedness
of the function $\lambda\mapsto\|\F_0(\lambda)\|_{\B(\H^s_t,\HS)}$.
\end{proof}

The next necessary result concerns the limits
$T(\lambda\pm i0):=\lim_{\varepsilon\searrow0}T(\lambda\pm i\varepsilon)$,
$\lambda\in\R_+$. Fortunately, it is already known (see for example
\cite[Lemma 9.1]{JK79}) that if $\sigma>1$ in \eqref{condV} then the limit
$
\big(1+R_0(\lambda+i0)V\big)^{-1}
:=\lim_{\varepsilon\searrow0}\big(1+R_0(\lambda+i\varepsilon)V\big)^{-1}
$
exists in $\B(\H_{-t},\H_{-t})$ for any $t\in(1/2,\sigma-1/2)$, and that the map
$\R_+\ni\lambda\mapsto\big(1+R_0(\lambda+i0)V\big)^{-1}\in\B(\H_{-t},\H_{-t})$ is
continuous. Corresponding results for $T(\lambda+i\varepsilon)$ follow immediately.
Note that only the limits from the upper half-plane have been computed in
\cite{JK79}, even though similar results for $T(\lambda-i0)$ could have been derived.
Due to this lack of information in the literature and for the simplicity of the
exposition, we consider from now on only the wave operator $W_-$.

\begin{Lemma}\label{lem_on_sigma}
Take $\sigma>5$ in \eqref{condV} and let $t\in(5/2,\sigma-5/2)$. Then, the function
$$
\R_+\ni\lambda\mapsto
\lambda^{1/4}\;\!T(\lambda+i0)\F_0(\lambda)^*\in\B(\HS,\H_{\sigma-t})
$$
is continuous and bounded, and the multiplication operator
$B:C_{\rm c}\big(\R_+;\HS\big)\to\ltwo(\R_+;\H_{\sigma-t})$ given by
\begin{equation}\label{defdeB}
(B\;\!\varphi)(\lambda)
:=\lambda^{1/4}\;\!T(\lambda+i0)\F_0(\lambda)^*\varphi(\lambda)\in\H_{\sigma-t},
\quad\varphi\in C_{\rm c}\big(\R_+;\HS\big),~\lambda\in\R_+,
\end{equation}
extends to an element of $\B\big(\Hrond,\ltwo(\R_+;\H_{\sigma-t})\big)$.
\end{Lemma}

\begin{proof}
The continuity of the function
$\lambda\mapsto\lambda^{1/4}\;\!T(\lambda+i0)\F_0(\lambda)^*\in\B(\HS,\H_{\sigma-t})$
follows from what has been said before. For the boundedness, it is sufficient to show
that the function
\begin{equation}\label{amontrer}
\R_+\ni\lambda\mapsto
\lambda^{1/4}\big\|\;\!T(\lambda+i0)\F_0(\lambda)^*\big\|_{\B(\HS,\H_{\sigma-t})}
\end{equation}
is bounded in a neighbourhood of $0$ and in a neighbourhood of $+\infty$.

For $\lambda>1$, we know from \cite[Lemma~9.1]{JK79} that the function
$\lambda\mapsto\|T(\lambda+i0)\|_{\B(\H_{-t},\H_{\sigma-t})}$ is bounded. We also
know from Lemma \ref{lem1} that the function
$\R_+\ni\lambda\mapsto\lambda^{1/4}\|\F_0(\lambda)^*\|_{\B(\HS,\H_{-t})}$ is bounded.
Thus, the function \eqref{amontrer} stays bounded in a neighbourhood of $+\infty$.

For $\lambda$ in a neighbourhood of $0$, we use asymptotic developments for
$T(\lambda+i0)$ and $\F_0(\lambda)^*$. The development for $\F_0(\lambda)^*$ (to be
found in \cite[Sec.~5]{JK79}) can be written as follows. For each $s\in\R$, there
exist $\gamma_0^*,\gamma_1^*\in\B(\HS,\H^s_{-t})$ such that
$$
\textstyle
\F_0(\lambda)^*=\big(\frac\lambda4\big)^{1/4}
\big(\gamma_0^*-i\lambda^{1/2}\gamma_1^*+o(\lambda^{1/2})\big)
\quad\hbox{in}\quad\B\big(\HS,\H^s_{-t}\big)\hbox{ as }\lambda\searrow 0.
$$
The development for $T(\lambda+ i0)$ as $\lambda\searrow 0$ has been computed in
\cite[Lemmas~4.1 to~4.5]{JK79}. It varies drastically depending on the presence of
$0$-energy eigenvalue and/or $0$-energy resonance. We reproduce here the most
singular behavior possible (\cf \cite[Lemma~4.5]{JK79})\;\!:
$$
T(\lambda+i0)=\lambda^{-1}VP_0V-i\lambda^{-1/2}\;\!C+O(1)
\quad\hbox{in}\quad\B(\H^1_{-t};\H_{\sigma-t})\hbox{ as }\lambda\searrow0,
$$
with $P_0$ the orthogonal projection onto $\ker(H)$ and
$C\in\B(\H^1_{-t};\H_{\sigma-t})$. Now, using these expressions for $\F_0(\lambda)^*$
and $T(\lambda+i0)$, one can write $\lambda^{1/4}T(\lambda+i0)\F_0(\lambda)^*$ as a
sum of terms bounded in $\B(\HS,\H_{\sigma-t})$ as $\lambda\searrow 0$ plus a term
$\frac1{\sqrt2}\lambda^{-1/2}VP_0V\gamma_0^*$ which is apparently unbounded. However,
we know from the proof of \cite[Thm.~5.3]{JK79} that $P_0V\gamma_0^*=0$. Thus, all
the terms in the asymptotic development of
$\lambda^{1/4}\;\!T(\lambda+i0)\F_0(\lambda)^*$ are effectively bounded in
$\B(\HS,\H_{\sigma-t})$ as $\lambda\searrow 0$, and thus the claim about boundedness
is proved. The claim  on the operator $B$ is then a simple consequence of what
precedes.
\end{proof}

\begin{Remark}
If one assumes that $H$ has no $0$-energy eigenvalue and/or no $0$-energy resonance,
then one can prove Lemma \ref{lem_on_sigma} under a weaker assumption on the decay of
$V$ at infinity. However, even if the absence of $0$-energy eigenvalue and $0$-energy
resonance is generic, we do not want to make such an implicit assumption in the
sequel. The condition on $V$ is thus imposed adequately.
\end{Remark}

Before deriving our main result, we recall the action of the dilation group
$\{U^+_\tau\}_{\tau\in\R}$ in $\ltwo(\R_+)$, namely,
$$
\big(U^+_\tau f\big)(\lambda):=\e^{\tau/2}f(\e^\tau\lambda),
\quad f\in C_{\rm c}(\R_+),~\lambda\in\R_+,~\tau\in\R,
$$
and denote its self-adjoint generator by $A_+$. We also introduce the
function $\vartheta\in C(\R)\cap\linf(\R)$ given by
\begin{equation}\label{defvar}
\vartheta(\nu):=\frac12\big(1-\tanh(2\pi\nu)-i\cosh(2\pi\nu)^{-1}\big),\quad\nu\in\R.
\end{equation}
Finally, we recall that the Hilbert spaces $\ltwo(\R_+;\H^s_t)$ and $\Hrond$ can be
naturally identified with the Hilbert spaces $\ltwo(\R_+)\otimes\H^s_t$ and
$\ltwo(\R_+)\otimes\HS$.

\begin{Theorem}\label{BigMama}
Take $\sigma>7$ in \eqref{condV} and let $\,t\in(7/2,\sigma-7/2)$. Then, one has in
$\B(\Hrond)$ the equality
\begin{equation}\label{ademontrer}
\F_0(W_--1)\;\!\F_0^*
=-2\pi i\;\!M\;\!\big\{\vartheta(A_+)\otimes1_{\H_{\sigma-t}}\big\}B,
\end{equation}
with $M$ and $B$ defined in \eqref{defdeM} and \eqref{defdeB}.
\end{Theorem}

The proof below consists in two parts. First, we show that the expression
\eqref{start} is well-defined for $\varphi$ and $\psi$ in dense subsets of $\Hrond$
(and thus equal to
$\big\langle\;\!\F_0(W_\pm-1)\;\!\F_0^*\varphi,\psi\big\rangle_{\!\Hrond}$ due to the
computations presented at the beginning of the section). Second, we show that the
expression \eqref{start} is equal to
$
\big\langle-2\pi i\;\!M\;\!
\big\{\vartheta(A_+)\otimes1_{\H_{\sigma-t}}\big\}B\varphi,\psi\big\rangle_{\!\Hrond}
$.

\begin{proof}
Take $\varphi\in C_{\rm c}(\R_+;\HS)$ and
$\psi\in C_{\rm c}^\infty(\R_+)\odot C(\S^2)$, and set $s:=\sigma-t>7/2$. Then, we
have for each $\varepsilon>0$ and $\lambda\in\R_+$ the inclusions
$$
g_\varepsilon(\lambda):=\lambda^{1/4}\;\!T(\lambda+i\varepsilon)\;\!\F_0^*\;\!
\delta_\varepsilon(L-\lambda)\varphi\in\H_s
\qquad\hbox{and}\qquad
f(\lambda):=\lambda^{-1/4}\F_0(\lambda)^*\psi(\lambda)\in\H_{-s}\;\!.
$$
It follows that the expression \eqref{start} is equal to
\begin{align*}
&-\int_\R\d\lambda\,\lim_{\varepsilon\searrow0}\int_0^\infty\d\mu\,\big\langle
T(\lambda+i\varepsilon)\;\!\F_0^*\;\!\delta_\varepsilon(L-\lambda)\varphi,
(\mu-\lambda+i\varepsilon)^{-1}\;\!\F_0(\mu)^*\;\!\psi(\mu)
\big\rangle_{\H_s,\H_{-s}}\\
&=-\int_{\R_+}\d\lambda\,\lim_{\varepsilon\searrow0}\int_0^\infty\d\mu\,
\bigg\langle g_\varepsilon(\lambda),\frac{\lambda^{-1/4}\mu^{1/4}}
{\mu-\lambda+i\varepsilon}\;\!f(\mu)\bigg\rangle_{\H_s,\H_{-s}}.
\end{align*}

Now, using the formula
$
(\mu-\lambda+i\varepsilon)^{-1}
=-i\int_0^\infty\d z\e^{i(\mu-\lambda)z}\e^{-\varepsilon z}
$
and then applying Fubini's theorem, one obtains that
\begin{align}
&\lim_{\varepsilon\searrow0}\int_0^\infty\d\mu\,\bigg\langle g_\varepsilon(\lambda),
\frac{\lambda^{-1/4}\mu^{1/4}}{\mu-\lambda+i\varepsilon}\;\!f(\mu)
\bigg\rangle_{\H_s,\H_{-s}}\nonumber\\
&=-i\lim_{\varepsilon\searrow0}\int_0^\infty\d z\,\e^{-\varepsilon z}
\bigg\langle g_\varepsilon(\lambda),\int_0^\infty\d\mu\,\e^{i(\mu-\lambda)z}
\lambda^{-1/4}\mu^{1/4}f(\mu)\bigg\rangle_{\H_s,\H_{-s}}\nonumber\\
&=-i\lim_{\varepsilon\searrow0}\int_0^\infty\d z\,\e^{-\varepsilon z}
\bigg\langle g_\varepsilon(\lambda),\int_{-\lambda}^\infty\d\nu\,\e^{i\nu z}
\left(\frac{\nu+\lambda}{\lambda}\right)^{1/4}f(\nu+\lambda)
\bigg\rangle_{\H_s,\H_{-s}}.\label{eq2}
\end{align}
Furthermore, the integrant in \eqref{eq2} can be bounded independently of
$\varepsilon\in(0,1)$. Indeed, one has
\begin{align}
&\left|\,\e^{-\varepsilon z}\bigg\langle g_\varepsilon(\lambda),
\int_{-\lambda}^\infty\d\nu\,\e^{i\nu z}
\left(\frac{\nu+\lambda}{\lambda}\right)^{1/4}f(\nu+\lambda)
\bigg\rangle_{\H_s,\H_{-s}}\right|\nonumber\\
&\le\big\|g_\varepsilon(\lambda)\big\|_{\H_s}\;\!\bigg\|\int_{-\lambda}^\infty\d\nu\,
\e^{i\nu z}\left(\frac{\nu+\lambda}{\lambda}\right)^{1/4}
f(\nu+\lambda)\bigg\|_{\H_{-s}},\label{dure}
\end{align}
and we know from Lemma \ref{lemlimite} and the paragraph following it that
$g_\varepsilon(\lambda)$ converges to
$
g_0(\lambda):=\lambda^{1/4}\;\!T(\lambda+i0)\F_0^*(\lambda)\varphi(\lambda)
$
in $\H_s$ as $\varepsilon\searrow0$. Therefore, the family
$\|\;\!g_\varepsilon(\lambda)\|_{\H_s}$ (and thus the r.h.s. of \eqref{dure}) is
bounded by a constant independent of $\varepsilon\in(0,1)$.

In order to exchange the integral over $z$ and the limit $\varepsilon\searrow0$ in
\eqref{eq2}, it remains to show that the second factor in \eqref{dure} belongs to
$\lone(\R_+,\d z)$. For that purpose, we denote by $h_\lambda$ the trivial extension
of the function
$
(-\lambda,\infty)\ni\nu\mapsto
\big(\frac{\nu+\lambda}{\lambda}\big)^{1/4}f(\nu+\lambda)\in\H_{-s}
$
to all of $\R$, and then note that the second factor in \eqref{dure} can be rewritten
as $(2\pi)^{1/2}\|(\F_1^*h_\lambda)(z)\|_{\H_{-s}}$, with $\F_1$ the one-dimensional
Fourier transform. To estimate this factor, observe that if $P_1$ denotes the
self-adjoint operator $-i\nabla$ on $\R$, then
$$
\big\|\big(\F_1^*h_\lambda\big)(z)\big\|_{\H_{-s}}
=\langle z\rangle^{-2}
\big\|\big(\F_1^*\langle P_1\rangle^2h_\lambda\big)(z)\big\|_{\H_{-s}},
\quad z\in\R_+\;\!.
$$
Consequently, one would have that
$
\|(\F_1^*h_\lambda)(z)\|_{\H_{-s}}\in\lone(\R_+,\d z)
$
if the norm
$
\big\|\big(\F_1^*\langle P_1\rangle^2h_\lambda\big)(z)\big\|_{\H_{-s}}
$
were bounded independently of $z$. Now, if $\psi=\eta\otimes\xi$ with
$\eta\in C_{\rm c}^\infty(\R_+)$ and $\xi\in C(\S^2)$, then one has for any
$x\in\R^3$
$$
\big(f(\nu+\lambda)\big)(x)
=\frac1{4\pi^{3/2}}\;\!\eta(\nu+\lambda)\int_{\S^2}\d\omega\,
\e^{i\sqrt{\nu+\lambda}\;\!\omega\cdot x}\xi(\omega).
$$
Therefore, one has
\begin{equation}\label{hahaha}
\big(h_\lambda(\nu)\big)(x)=
\begin{cases}
\frac1{4\pi^{3/2}}\left(\frac{\nu+\lambda}{\lambda}\right)^{1/4}\eta(\nu+\lambda)
\int_{\S^2}\d\omega\,\e^{i\sqrt{\nu+\lambda}\;\!\omega\cdot x}\xi(\omega)
& \nu>-\lambda\\
0 & \nu\le-\lambda,
\end{cases}
\end{equation}
which in turns implies that
$$
\big|\big\{\big(\F_1^*\langle P_1\rangle^2h_\lambda\big)(z)\big\}(x)\big|
\le{\rm Const.}\;\!\langle x\rangle^2,
$$
with a constant independent of $x\in\R^3$ and $z\in\R_+$. Since the r.h.s. belongs to
$\H_{-s}$ for $s>7/2$, one concludes that
$\big\|\big(\F_1^*\langle P_1\rangle^2h_\lambda\big)(z)\big\|_{\H_{-s}}$ is bounded
independently of $z$ for each $\psi=\eta\otimes\xi$, and thus for each
$\psi\in C^\infty_{\rm c}(\R_+)\odot C(\S^2)$ by linearity. As a consequence, one can
apply Lebesgue dominated convergence theorem and obtain that \eqref{eq2} is equal to
$$
-i\;\!\bigg\langle g_0(\lambda),\int_0^\infty\d z\int_\R\d\nu\;\!
\e^{i\nu z}h_\lambda(\nu)\bigg\rangle_{\H_s,\H_{-s}}.
$$

With this equality, one has concluded the first part of the proof; that is, one has
justified the equality between the expression \eqref{start} and
$\big\langle\;\!\F_0(W_\pm-1)\;\!\F_0^*\varphi,\psi\big\rangle_{\!\Hrond}$ on the
dense sets of vectors introduced at the beginning of the proof.

The next task is to show that
$
\big\langle\;\!\F_0(W_\pm-1)\;\!\F_0^*\varphi,\psi\big\rangle_{\!\Hrond}
$
is equal to
$
\big\langle-2\pi i\;\!M\big\{\vartheta(A_+)\otimes1_{\H_{-s}}\big\}B\varphi,
\psi\big\rangle_{\!\Hrond}
$.
For that purpose, we write $\chi_+$ for the characteristic function for $\R_+$. Since
$h_\lambda$ has compact support, we obtain the following equalities in the sense of
distributions (with values in $\H_{-s}$)\;\!:
\begin{align*}
\int_0^\infty\d z\int_\R\d\nu\;\!\e^{i\nu z}h_\lambda(\nu)
&=\sqrt{2\pi}\int_\R\d\nu\,\big(\F_1^*\chi_+\big)(\nu)\;\!h_\lambda(\nu)\\
&=\sqrt{2\pi}\int_{-\lambda}^\infty\d\nu\,\big(\F_1^*\chi_+\big)(\nu)
\left(\frac{\nu+\lambda}{\lambda}\right)^{1/4}f(\nu+\lambda)\\
&=\sqrt{2\pi}\int_\R\d\mu\,\big(\F_1^*\chi_+\big)\big(\lambda(\e^\mu-1)\big)
\;\!\lambda\e^{5\mu/4}f(\e^\mu\lambda)\qquad(\e^\mu\lambda:=\nu+\lambda)\\
&=\sqrt{2\pi}\int_\R\d\mu\,\big(\F_1^*\chi_+\big)\big(\lambda(\e^\mu-1)\big)
\;\!\lambda\e^{3\mu/4}\big\{\big(U_\mu^+\otimes1_{\H_{-s}}\big)f\big\}(\lambda).
\end{align*}
Then, by using the fact that
$
\F_1^*\chi_+
=\sqrt{\frac\pi2}\;\!\delta_0+\frac i{\sqrt{2\pi}}\;\!\Pv\frac1{(\;\!\cdot\;\!)}
$
with $\delta_0$ the Dirac delta distribution and $\Pv$ the principal value, one gets
that
$$
\int_0^\infty\d z\int_\R\d\nu\;\!\e^{i\nu z}h_\lambda(\nu)
=\int_\R\d\mu\left(\pi\;\!\delta_0(\e^\mu-1)
+i\;\!\Pv\frac{\e^{3\mu/4}}{\e^\mu-1}\right)
\big\{\big(U_\mu^+\otimes1_{\H_{-s}}\big)f\big\}(\lambda).
$$
So, by considering the identity
$$
\frac{\e^{3\mu/4}}{\e^\mu-1}
=\frac14\left(\frac1{\sinh(\mu/4)}+\frac1{\cosh(\mu/4)}\right)
$$
and the equality \cite[Table 20.1]{Jef95}
\begin{equation*}
\big(\F_1\bar\vartheta\big)(\nu)
:=\sqrt{\frac\pi2}\;\!\delta_0\big(\e^\nu-1\big)
+\frac i{4\sqrt{2\pi}}\;\!\Pv\left(\frac1{\sinh(\nu/4)}+\frac1{\cosh(\nu/4)}\right),
\end{equation*}
with $\vartheta$ defined in \eqref{defvar}, one infers that
\begin{align*}
&\big\langle\F_0(W_--1)\;\!\F_0^*\varphi,\psi\big\rangle_{\!\Hrond}\\
&=i\int_{\R_+}\d\lambda\,\bigg\langle g_0(\lambda),\int_\R\d\mu\,
\bigg\{\pi\;\!\delta_0\big(\e^\mu-1\big)\\
&\hspace{120pt}+\frac i4\;\!\Pv\left(\frac1{\sinh(\mu/4)}
+\frac1{\cosh(\mu/4)}\right)\bigg\}
\big\{\big(U_\mu^+\otimes1_{\H_{-s}}\big)f\big\}(\lambda)
\bigg\rangle_{\H_s,\H_{-s}}\\
&=i\sqrt{2\pi}\int_{\R_+}\d\lambda\left\langle g_0(\lambda),\int_\R\d\mu\,
\big(\F_1\bar\vartheta\big)(\mu)\;\!
\big\{\big(U_\mu^+\otimes1_{\H_{-s}}\big)f\big\}(\lambda)
\right\rangle_{\H_s,\H_{-s}}.
\end{align*}
Finally, by recalling that
$
\big\{\vartheta(A_+)\otimes1_{\H_{-s}}\big\}f
=\frac1{\sqrt{2\pi}}\int_\R\d\mu\,\big(\F_1\bar\vartheta\big)(\mu)
\big(U_\mu^+\otimes 1_{\H_{-s}}\big)f
$,
that $g_0(\lambda)=(B\varphi)(\lambda)$ and that $f=M^*\psi$, one obtains
\begin{align*}
\big\langle\F_0(W_--1)\F_0^*\varphi,\psi\big\rangle_{\!\Hrond}
&=2\pi i\int_{\R_+}\d\lambda\,\big\langle(B\varphi)(\lambda),
\big\{\big(\vartheta(A_+)^*\otimes1_{\H_{-s}}\big)M^*\psi\big\}(\lambda)
\big\rangle_{\H_s,\H_{-s}}\\
&=\big\langle-2\pi i\;\!M\;\!\big\{\vartheta(A_+)\otimes1_{\H_{-s}}\big\}B\varphi,
\psi\big\rangle_{\!\Hrond}.
\end{align*}
This concludes the proof, since the sets of vectors $\varphi\in C_{\rm c}(\R_+;\HS)$
and $\psi\in C_{\rm c}^\infty(\R_+)\odot C(\S^2)$ are dense in $\Hrond$.
\end{proof}

We now derive a technical lemma which will be essential for the proof of Theorem
\ref{Java}.

\begin{Lemma}\label{Lemma_compact}
Take $s>-1$ and $t>3/2$. Then, the difference
$$
\big\{\vartheta(A_+)\otimes1_{\HS}\big\}M
-M\big\{\vartheta(A_+)\otimes 1_{\H^s_t}\big\}
$$
belongs to $\K\big(\ltwo(\R_+;\H_t^s),\Hrond\big)$.
\end{Lemma}

\begin{proof}
(i) The unitary operator $\GG:\ltwo(\R)\to\ltwo(\R_+)$ given by
$$
(\GG f)(\lambda):=\lambda^{-1/2}f\big(\ln(\lambda)\big),
\quad f\in C^\infty_{\rm c}(\R),~\lambda\in\R_+,
$$
satisfies $(\GG^*\;\!U_\tau^+\;\!\GG f)(x)=f(x+\tau)$ and
$(\GG^*\e^{i\tau\ln(\LL)}\GG f)(x)=\e^{i\tau x}f(x)$ for each $x,\tau\in\R$, with
$\LL$ the maximal multiplication operator in $\ltwo(\R_+)$ by the variable in $\R_+$.
It follows that $\GG^*A_+\GG=P_1$ on $\dom(P_1)$ and that $\GG^*\ln(\LL)\;\!\GG=X_1$
on $\dom(X_1)$, with $P_1$ and $X_1$ the self-adjoint operators of momentum and
position in $\ltwo(\R)$.

Now, take $f_1,f_2$ two complex-valued continuous functions on $\R$ having limits at
$\pm\infty$; that is, $f_1,f_2\in C([-\infty,\infty])$. Then, a standard result of
Cordes implies the inclusion $[f_1(P_1),f_2(X_1)]\in\K\big(\ltwo(\R)\big)$ (see for
instance \cite[Thm.~4.1.10]{ABG}). Conjugating this inclusion with the unitary
operator $\GG$, one thus infers that
$\big[f_1(A_+),f_3(\LL)\big]\in\K\big(\ltwo(\R_+)\big)$ with
$f_3:=f_2\circ\ln\in C([0,\infty])$.

(ii) We know from Lemma \ref{lem2} and Definition \eqref{defdeM} that
$$
(M\xi)(\lambda):=m(\lambda)\;\!\xi(\lambda),
\quad\xi\in C_{\rm c}(\R_+;\H^s_t),~\lambda\in\R_+,
$$
with $m\in C\big([0,\infty];\K(\H^s_t,\HS)\big)$. We also know that the algebraic
tensor product $C([0,\infty])\odot\K(\H^s_t,\HS)$ is dense in
$C\big([0,\infty];\K(\H^s_t,\HS)\big)$, when $C\big([0,\infty];\K(\H^s_t,\HS)\big)$
is equipped with the uniform topology (see \cite[Thm.~1.15]{Pro77}). So, for each
$\varepsilon>0$ there exist $n\in\N^*$, $a_j\in C([0,\infty])$ and
$b_j\in\K(\H^s_t,\HS)$ such that such that
$
\big\|M-\sum_{j=1}^n a_j(\LL)\otimes b_j\big\|_{\B(\ltwo(\R_+;\H^s_t),\Hrond)}
<\varepsilon.
$
Therefore, in order to prove the claim, it is sufficient to show that the operator
\begin{equation}\label{diff_Cn}
\big\{\vartheta(A_+)\otimes1_{\HS}\big\}\Bigg\{\sum_{j=1}^na_j(\LL)\otimes b_j\Bigg\}
-\Bigg\{\sum_{j=1}^na_j(\LL)\otimes b_j\Bigg\}
\big\{\vartheta(A_+)\otimes 1_{\H^s_t}\big\}
=\sum_{j=1}^n\big[\vartheta(A_+),a_j(\LL)\big]\otimes b_j
\end{equation}
is compact. But, we know that $b_j\in \K(\H^s_t,\HS)$ and that
$\big[\vartheta(A_+),a_j(\LL)\big]\in\K\big(\ltwo(\R_+)\big)$ due to point (i). So,
it immediately follows that the operator \eqref{diff_Cn} is compact, since finite
sums and tensor products of compact operators are compact operators (see
\cite[Thm.~2]{Hol72}).
\end{proof}

Before giving the proof of Theorem \ref{Java}, we recall the action of the dilation
group $\{U_\tau\}_{\tau\in\R}$ in $\H$, namely,
$$
\big(U_\tau f\big)(x):=\e^{3\tau/2}f(\e^\tau x),
\quad f\in C_{\rm c}(\R^3),~x\in\R^3,~\tau\in\R,
$$
and denote its self-adjoint generator by $A$. The image $\F_0R(A)\;\!\F_0^*$ of
$R(A):=\frac12\big(1+\tanh(\pi A)-i\cosh(\pi A)^{-1}\big)$ in $\B(\Hrond)$ can be
easily computed. Indeed, one has the decomposition $\F_0=\U\F$, with $\U:\H\to\Hrond$
given by
$
\big((\U f)(\lambda)\big)(\omega)
:=\big(\frac\lambda4\big)^{1/4}f(\sqrt\lambda\;\!\omega)
$
for each $f\in\SS$, $\lambda\in\R_+$ and $\omega\in\S^2$. Furthermore, one has the
identities $\F A\;\!\F^*=-A$ on $\dom(A)$ and $\U A\;\!\U^*=2A_+\otimes 1_{\HS}$ on
$\dom(A_+\otimes1_{\HS})$. Therefore, one obtains that
$$
\F_0R(A)\!\;\F_0^*=\vartheta(A_+)\otimes 1_{\HS}.
$$

\begin{proof}[Proof of Theorem \ref{Java}]
Set $s=0$ and $t\in(7/2,\sigma-7/2)$. Then, we deduce from Theorem \ref{BigMama},
Lemma \ref{Lemma_compact} and the above paragraph that
\begin{align*}
W_--1
&=-2\pi i\;\!\F_0^*M\big\{\vartheta(A_+)\otimes1_{\H_{\sigma-t}}\big\}B\;\!\F_0\\
&=-2\pi i\;\!\F_0^*\big\{\vartheta(A_+)\otimes 1_\HS\big\}MB\;\!\F_0+K\\
&=R(A)\;\!\F_0^*(-2\pi iMB)\;\!\F_0+K,
\end{align*}
with $K\in\K(\H)$. Comparing $-2\pi iMB$ with the usual expression for the scattering
matrix $S(\lambda)$ (see for example \cite[Eq.~(5.1)]{JK79}), one observes that
$-2\pi iMB=\int_{\R_+}^\oplus\d\lambda\,\big(S(\lambda)-1\big)$. Since $\F_0$ defines
the spectral representation of $H_0$, one obtains that
\begin{equation}\label{form_BPU}
W_--1=R(A)(S-1)+K.
\end{equation}
The formula for $W_+-1$ follows then from \eqref{form_BPU} and the relation
$W_+=W_-\;\!S^*$.
\end{proof}

\begin{Remark}\label{Rem_comments}
Formulas \eqref{jolieformule} were already obtained in \cite{KR12} under an implicit
assumption. The only difference is that the operator $R(A)$ is replaced in
\cite{KR12} by an operator $\boldsymbol\varphi(A)$ slightly more complicated. The
resulting formulas for the wave operators differ by a compact term, but compact
operators do not play any role in the algebraic construction (both expressions for
the wave operators belong to the $C^*$-algebra constructed in \cite[Sec.~4]{KR12} and
thus coincide after taking the quotient by the ideal of compact operators).
Consequently, the topological approach of Levinson's theorem presented in \cite{KR12}
also applies here, with the implicit assumption no longer necessary.
\end{Remark}



\end{document}